\documentclass{lmcs}
\pdfoutput=1


\usepackage{todonotes}
\usepackage{cite}
\usepackage{tikz}
\usetikzlibrary{cd}
\usepackage{hyperref}
\usepackage{tikzit}

\tikzstyle{Black}=[fill=black, draw=black, shape=circle, inner sep=0, minimum size=6pt]
\tikzstyle{Small box}=[fill=white, draw=black, shape=chamfered rectangle, minimum size=8, chamfered rectangle corners=north east, chamfered rectangle xsep=2, chamfered rectangle ysep=2, tikzit shape=rectangle]
\tikzstyle{Small opbox}=[fill=white, draw=black, shape=chamfered rectangle, tikzit draw=white, tikzit fill={rgb,255: red,191; green,191; blue,191}, minimum size=8, chamfered rectangle corners=north west, chamfered rectangle xsep=2, chamfered rectangle ysep=2, tikzit shape=rectangle]
\tikzstyle{Red}=[fill=red, draw=black, shape=circle]
\tikzstyle{map small}=[fill=white, draw=black, shape=rounded rectangle, minimum width=23, minimum height=16, rounded rectangle west arc=none, tikzit shape=rectangle, tikzit draw=blue]
\tikzstyle{comap small}=[fill=white, draw=black, shape=rounded rectangle, minimum width=23, minimum height=16, rounded rectangle east arc=none, tikzit shape=rectangle, tikzit draw=red]
\tikzstyle{Medium box}=[fill=white, draw=black, shape=chamfered rectangle, minimum size=20, chamfered rectangle corners=north east, chamfered rectangle xsep=2, chamfered rectangle ysep=2, tikzit shape=rectangle]
\tikzstyle{Empty}=[fill=white, draw=black, shape=rectangle, dashed, minimum size=16, tikzit fill={rgb,255: red,191; green,191; blue,191}, tikzit draw=black]
\tikzstyle{White}=[fill=white, draw=black, shape=circle, inner sep=0, minimum size=6pt]
\tikzstyle{Vertex}=[fill=black, draw=black, shape=circle, inner sep=0, minimum size=3pt]
\tikzstyle{Edge}=[fill=white, draw=black, shape=rectangle, minimum size=6]

\tikzstyle{dashedge}=[dashed, -]
\tikzstyle{dasharrow}=[->, dashed]
\tikzstyle{twoheadarrow}=[<->]

\newcommand{\from}{\colon}
\newcommand{\seq}{\ensuremath{\mathrel{;}}}

\DeclareMathOperator{\Hom}{Hom}
\DeclareMathOperator{\Set}{\mathsf{Set}}
\DeclareMathOperator{\FinSet}{\mathsf{FinSet}}
\newcommand{\Relc}[1]{\Rel(#1)}
\DeclareMathOperator{\Rel}{\mathsf{Rel}}
\DeclareMathOperator{\LinRel}{\mathsf{LinRel}}
\DeclareMathOperator{\id}{id}
\newcommand{\MorphTame}{F}

\newcommand{\Bicat}{\mathcal{B}}
\newcommand{\Cat}{\mathcal{C}}
\newcommand{\Span}[1]{\mathsf{Span} \, #1}
\newcommand{\Spanleq}[1]{\mathsf{Span}^{\leq} #1}
\newcommand{\Spantilde}[1]{\mathsf{Span}^{\sim} #1}
\newcommand{\SpanS}[1]{\mathsf{Span}^{\mathcal{S}} #1}
\newcommand{\SpanI}[1]{\mathsf{Span}^{\mathcal{I}} #1}
\newcommand{\Cospantilde}[1]{\mathsf{Cospan}^{\sim} #1}
\newcommand{\mc}[1]{\mathcal{#1}}
\newcommand{\op}[1]{#1^{\opposite}}
\DeclareMathOperator{\opposite}{op}
\DeclareMathOperator{\Map}{Map}
\DeclareMathOperator{\Comap}{Comap}

\newcommand{\scale}[1]{\scalebox{0.7}{#1}}
\newcommand{\scalex}[1]{\scalebox{0.6}{#1}}
\newcommand{\stikzfig}[1]{\scale{\input{#1.tikz}}}
\newcommand{\sxtikzfig}[1]{\scalex{\input{#1.tikz}}}

\begin{document}

\title{Cartesian bicategories with choice}

\author[F.~Bonchi]{Filippo Bonchi}
\address{University of Pisa}
\thanks{This work is
supported by the Ministero dell'Universit\`a e della Ricerca Scientifica 
of Italy under Grant No.\ 201784YSZ5, PRIN2017 -- ASPRA
(\emph{Analysis of Program Analyses})}

\author[J.~Seeber]{Jens Seeber}
\address{University of Pisa}

\author[P.~Soboci{\'n}ski]{Pawe{\l} Soboci{\'n}ski}
\address{Tallinn University of Technology}
\thanks{Supported 
by the ESF funded Estonian IT Academy research measure
(project 2014-2020.4.05.19-0001).}

\begin{abstract}
Relational structures are emerging as ubiquitous mathematical machinery
in the semantics of open systems of various kinds.
Cartesian bicategories are a well-known categorical algebra of relations that has proved especially useful in recent applications.

The passage between a category and its bicategory of relations is an important question that has been widely studied for decades. We study an alternative construction that yields a cartesian bicategory of relations. Its behaviour is closely related to the axiom of choice, which itself can be expressed in the language of cartesian bicategories.
\end{abstract}

\maketitle

\section*{Introduction}
Cartesian bicategories (of relations) were introduced in~\cite{carboni1987cartesian} as a categorical algebra of relations,
and as an alternative to Freyd and Scedrov's allegories~\cite{freyd1990categories}. RFC Walters had a certain distaste for the approach through allegories;
he referred to the modular law of allegories as a \emph{formica mentale}, a
``complication which prevents thought''~\cite{formica}.

In recent years
cartesian bicategories
have received renewed attention by researchers interested in string-diagrammatic languages. Indeed, thanks to the compact closed structure induced by Frobenius bimonoids, cartesian bicategories have proved to be a powerful theoretical framework in the compositional studies of different kinds of feedback systems. For instance, signal flow graphs \cite{mason1953feedback}---circuit-like specifications of linear dynamical systems---form a cartesian bicategory~\cite{bonchi2017refinement}. Moreover, the fact that cartesianity only holds laxly makes them able to serve as ``resource-sensitive'' syntax, as outlined in~\cite{relationaltheories}, where free cartesian bicategories were proposed as a resource-sensitive generalisation of Lawvere theories.

Free cartesian bicategories were also used in~\cite{GCQ}, where we showed that their algebraic presentation can be seen as an equational characterisation of well-known logical
preorders, namely those arising from query inclusion of conjunctive queries (aka regular logic). The deep relationship between cartesian bicategories
and regular logic---already alluded to in~\cite{carboni1987cartesian}---was also recently touched upon by Fong and Spivak~\cite{fong2018graphical}.

In cartesian bicategories, it is important to distinguish between arbitrary morphisms---which can be thought of as relations---and a certain class of morphisms called \emph{maps}, which can be thought of as functions. A fundamental result~\cite[Theorem~3.5]{carboni1987cartesian} states that, for a cartesian bicategory $\mathcal{B}$ satisfying the property of \emph{functional completeness}, (i) the subcategory of maps (denoted by $\Map{\mathcal{B}}$) is regular and (ii) the category of relations over the category of maps ($\Relc{\Map{B}}$) is biequivalent to $\mathcal{B}$.
Unfortunately, this beautiful result is not relevant for free cartesian bicategories: for instance the categories obtained by the algebraic presentations in~\cite{relationaltheories} and~\cite{GCQ} do not arise from the $\Relc{\cdot}$ construction.

For this reason in~\cite{GCQ}, we needed to rely on an alternative construction $\Spantilde{}$ that we believe is of independent interest.
The construction has previously appeared in the literature~\cite{guitart1980relations}, but has thus far not received the attention that it merits.
 First, it requires less structure of the underlying category: while $\Relc{\cdot}$ requires a regular category, $\Spantilde{}$ requires merely the presence of weak pullbacks,
which satisfy the existence clause in the universal property of pullbacks, but not necessarily the uniqueness clause. Second, while in the category of sets and functions both constructions yield the usual category of relations, as we shall see, there are important cases in which they differ.

\medskip

Our first main contribution is an analogue of the aforementioned result for $\Spantilde{}$, namely that $\Spantilde{\Map{\mathcal{B}}}$ is biequivalent to $\mathcal{B}$. In this setting, Carboni and Walters' functional completeness can be relaxed to a weaker condition that we call \emph{having enough maps}, but an additional assumption is necessary: $\mathcal{B}$ has to satisfy
the \emph{axiom of choice}. Indeed, our first main result (Theorem~\ref{thm:Choice}) asserts that a cartesian bicategory $\mathcal{B}$ with enough maps satisfies the axiom of choice \emph{if and only if} $\mathcal{B}$ is biequivalent to $\Spantilde{\Map{\mathcal{B}}}$.

This characterisation motivates a closer look at the axiom of choice,
one of the best known---and most controversial---axioms of set theory~\cite{herrlich2006axiom}.
It has many ZF-equivalent formulations, some requiring only very basic concepts. One is:
\begin{center}
  Every total relation contains a map.
\end{center}
We observe that this formulation is natural to state in the language of cartesian bicategories. Another way of viewing our result, therefore,
is that cartesian bicategories with enough maps, satisfying the axiom of choice are precisely those that arise via the $\Spantilde{}$ construction.

\medskip

Given the innovations of
topos theory~\cite{elephant} in foundations of mathematics, the question of whether or not to accept the axiom of choice
is nowadays less absolute (and therefore less heated). Indeed, if a topos is a mathematical ``universe'', then it
holds in some and not in others, thus accepting/rejecting choice turns from a philosophical question into a practical matter.
Interpreting choice inside a category does not need the full power of the internal language of a topos -- it suffices if the category in
question captures basic properties of relations. Cartesian bicategories can therefore be seen as an amusing setting for the study of the axiom of choice. Indeed,
the advantage of a weaker language is a finer grained analysis: e.g. we shall see that properties well-known to be equivalent to choice in ZF (e.g. surjective functions split) are different as properties of cartesian bicategories.

\medskip

Our second main contribution is the introduction of a generalisation of the $\Spantilde{}$ construction, that we call $\SpanS{}$. Here $\mathcal{S}$ is a \emph{systems of covers}, roughly a class of maps satisfying certain closure properties. As for $\Spantilde{}$, we identify necessary and sufficient conditions ensuring that
a cartesian bicategory $\Bicat$ can then be reconstructed as $\SpanS{\Map(\Bicat)}$ (Theorem \ref{thm:ReconstructTame}). These conditions are summarised in the notion of \emph{tame} cartesian bicategory.

Our interest in this novel construction is twofold: on the one hand, it allows for handling cartesian bicategories that are freely generated not only from a signature, like those in~\cite{GCQ}, but also from a set of equations, e.g.\  from a relational algebraic presentation~\cite{relationaltheories}. As an example, we show that the prop $\mathbf{ERel}$ of equivalence relations~\cite{ZanasiThesis,lmcs:4796,coya2017corelations,fong2016algebra,bruni2001some}, corresponding to the theory of non-empty sets~\cite{relationaltheories}, can be obtained as a $\SpanS{}$, but not as a  $\Spantilde{}$.

On the other hand, looking at $\SpanS{}$ allows us to give a simpler proof of Theorem~\ref{thm:Choice}. Indeed, Theorem \ref{thm:ReconstructTame} ensures that $\Bicat \cong \SpanS{\Map(\Bicat)}$ for $\mathcal{S}$ being the class of surjective maps. Observing that if surjectives split then $\SpanS{}$ and  $\Spantilde{}$ coincide is now enough to conclude the statement of Theorem~\ref{thm:Choice}.

\emph{Structure of the paper.}
In Section \ref{sec:Cartesian}, we use cartesian categories  as a convenient starter to introduce the string diagrammatic language and, at the same time, some notions relevant for cartesian bicategories. We provide an overview of a few important concepts of cartesian bicategories and their maps in Sections~\ref{sec:CartBicat} and \ref{sec:Maps}.
In Section~\ref{sec:AC} we define the axiom of choice in cartesian bicategories, the property of ``having enough maps'' and discuss ramifications of this, including
a useful characterisation. In Section~\ref{sec:Span} we introduce the $\mathsf{Span}^{\sim}$ construction and prove several results that are relevant for the proof of Theorem~\ref{thm:Choice}.
We introduce tame cartesian bicategories in Section~\ref{sec:CartBicatAndMaps} and the $\SpanS{}$ construction in Section~\ref{sec:spanS}.  There we also show Theorem~\ref{thm:ReconstructTame} and we use it to provide a proof for Theorem~\ref{thm:Choice}.
In Section \ref{sec:relatedWork}, we compare the constructions $\Rel(\Cat)$ and $\SpanS{\Cat}$ and show that they coincide when $\mathcal{S}$ is the class of regular epis. In particular, $\Rel(\Cat)$ and $\Spantilde{\Cat}$ coincide whenever regular epis split.

We would like to thank Aleks Kissinger and the team behind \href{https://tikzit.github.io}{TikZiT}, which was used to create the diagrams in this paper.

\section{Cartesian categories}\label{sec:Cartesian}

We will use \emph{string diagrams} as an intuitive graphical notation for what is formally represented as morphisms in a
symmetric monoidal category $\mathcal{C}$ with monoidal product $\otimes$ and monoidal unit $I$, details for this can be found in~\cite{selinger2010survey}.
Here we would like to give an intuitive way to read string diagrams, one that doesn't require the machinery of category theory.

Intuitively, $\stikzfig{f}$ denotes a process $f$ that receives an input of type $X$ and produces output of type $Y$. We will also write this as
$f \from X \to Y$.
It is possible to talk about several inputs and several outputs by stacking wires, for example $\stikzfig{Bigf}$ receives inputs of type $A$ and $B$
respectively and produces output of type $C$ and $D$ respectively.
In other words, a compound type is formed by stacking wires and we will write such compound type formed from types $A$ and $B$ as $A \otimes B$.

A special process is given by $\stikzfig{IdX}$, which is the process that doesn't do anything -- the identity on type $X$. We now have several ways
to build larger and more interesting processes:
For $f\colon X \to Y$ and $g\colon Y \to Z$, the composition $\stikzfig{fseqg}$, which in symbols we will also denote as $f \seq g \colon X \to Z$.
This is the process that first applies $f$ and then applies $g$ on the result.
For $f\colon X \to Y$ and $g\colon Z \to W$, their parallel composition is $\stikzfig{f+g}$, in symbols denoted as
$f \otimes g \colon X \otimes Z \to Y \otimes W$. This is the process that executes $f$ and $g$ in parallel.
It is now possible to iterate these ways of constructing diagrams to form processes of arbitrary complexity.
If we want to change the order of inputs or outputs, we can use the symmetries $\stikzfig{Sym}$.
As a special case, there is also a type $I$ that corresponds to having no wires and an empty diagram, denoted $\stikzfig{Empty}$ that represents the process
of doing nothing to no input and obtaining no output.

The formal theory of symmetric monoidal categories ensures that we do not need to worry about how our diagrams are constructed.
If two diagrams have the same connectivity, they represent the same process. This allows us to use diagrammatic reasoning, that is a formal
manipulation of diagrams, that now behave like a two-dimensional analogue of the terms used in algebra.

If it is clear from the context how wires are labelled, we will declutter our diagrams and omit labels.

We use this string diagrammatic language to introduce cartesian categories, which are those symmetric monoidal categories where it is possible to
copy and discard, which are processes that will be represented by $\stikzfig{Copy}$ and $\stikzfig{Dis}$ respectively.

\begin{defi}\label{def:cartCat}
  A cartesian category is a symmetric monoidal category $(\mathcal{B},\otimes, I)$, where
  every object $X \in \mathcal{B}$ is equipped with morphisms
  \[
  \stikzfig{CopyX} \from X \to X \otimes X \quad \text{and} \quad \stikzfig{DisX} \from X \to I
  \]
such that
\medskip
  \begin{enumerate}
    \item $\stikzfig{CopyX}$ and $\stikzfig{DisX}$ form a cocommutative comonoid, that is they satisfy
      \begin{align*}
       \sxtikzfig{Assoc} &
  \;\;\;   \;\;\; \;\;\; \;\; \sxtikzfig{Comm} &
       \sxtikzfig{Unit} \\
      \end{align*}
    \item Each morphism $f \from X \to Y$ is a comonoid homomorphism, that is
    \begin{align*}
        & \sxtikzfig{MapComult}
        & &\sxtikzfig{MapDis}
    \end{align*}
    \item The choice of comonoid on every object is compatible with the monoidal structure in the sense that
    \begin{align*}
      & \sxtikzfig{CohDis}
      & & \sxtikzfig{CohCopy}
    \end{align*}
    and
    \begin{align*}
        & \sxtikzfig{DisI}
      & & \sxtikzfig{CopyI}
    \end{align*}
  \end{enumerate}
\end{defi}

\begin{rem}
  In many interesting examples of monoidal categories, for example $\Set$ equipped with the usual Cartesian product, the monoidal structure is not strictly
  associative and unital but only up to natural isomorphisms called the associators and unitors.
  In $\Set$ these isomorphisms are necessary to move and remove parenthesis, for example to identify
  the tuple $((x,y),(z,())) \in (X \times Y) \times (Z \times 1)$ with the tuple $(x,(y,z)) \in X \times (Y \times Z)$
  which are different objects, despite the obvious similarity.
  Nevertheless, in string diagrams these natural isomorphisms do not feature and the justification for that lies in the coherence theorem for monoidal
  categories~\cite{maclane1963natural}, which says that every monoidal category is (monoidally) equivalent to a strict one. Therefore, string diagrams
  formally describe the strictification of a monoidal category, with the coherence theorem ensuring that nothing essential is lost in the process.
\end{rem}

\begin{exa}
  The most prominent example of a cartesian category is the category $\Set$ of sets and functions, equipped with the cartesian product as monoidal product.
  For each set $X$, the comonoid structure is given by the diagonal function $X\to X\times X$ and the unique function $X\to 1$.
\end{exa}

That the comonoid structure on every object is respected by every morphism induces a categorical product. In fact, this characterises cartesian categories
which is an observation first made in~\cite{fox_coalgebras_1976}.

\begin{prop}
  Cartesian categories are equivalently those symmetric monoidal categories where $\otimes$ is the categorical product and the monoidal unit $I$ is terminal.
  \label{prop:CartCopy}
\end{prop}
\begin{proof}
  \begin{itemize}
    \item Let $\otimes$ be the product and $I$ terminal. Let $\stikzfig{DisX}$ be the unique morphism $X \to I$ and let $\stikzfig{CopyX}$ be the diagonal
          morphism $\Delta \from X \to X \otimes X$, which is the unique morphism that makes the diagram
          
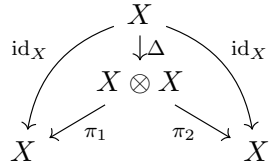
\[
            \begin{tikzcd}[sep=small]
              & X \arrow[swap, bend right]{ddl}{\id_X} \arrow{d}{\Delta} \arrow[bend left]{ddr}{\id_X} & \\
              & \arrow{dl}{\pi_1} X \otimes X \arrow[swap]{dr}{\pi_2} \\
              X & & X
            \end{tikzcd}
          \]
          commute, where $\pi_i$ are the projections out of the product. It is easy to check that $\stikzfig{DisX}$ and $\stikzfig{CopyX}$ satisfy the
          required axioms.
    \item Assume we have $\stikzfig{CopyX}$ and $\stikzfig{DisX}$ on every object $X$ as in Definition~\ref{def:cartCat}.
          Then $I$ is terminal because let $f \from X \to I$ be any morphism, then
          \[
            \stikzfig{Terminal}
          \]
          where we briefly made the no-wire type $I$ visible.
          That $X \otimes Y$ is the product of $X$ and $Y$ can be seen as follows:
          The projection $\pi_1 \from X \otimes Y \to X$ is given by \[\stikzfig{ProjX}\]
          and likewise the projection $\pi_2 \from X \otimes Y \to Y$ by \[\stikzfig{ProjY}\]
          Given morphisms $f \from T \to X$ and $g \from T \to Y$, consider the induced morphism $\alpha \from T \to X \otimes Y$ given by
          \[
            \stikzfig{Prodfg}
          \]
          It is straightforward to check that this makes the diagram
          
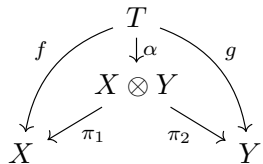
\[
            \begin{tikzcd}[sep=small]
              & T \arrow[swap, bend right]{ddl}{f} \arrow{d}{\alpha} \arrow[bend left]{ddr}{g} & \\
              & \arrow{dl}{\pi_1} X \otimes Y \arrow[swap]{dr}{\pi_2} \\
              X & & Y
            \end{tikzcd}
          \]
          commute. We will therefore prove uniqueness. Assume that there is $h \from T \to X \otimes Y$ making the same diagram
          
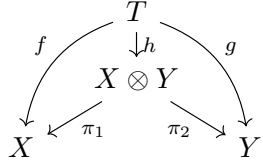
\[
            \begin{tikzcd}[sep=small]
              & T \arrow[swap, bend right]{ddl}{f} \arrow{d}{h} \arrow[bend left]{ddr}{g} & \\
              & \arrow{dl}{\pi_1} X \otimes Y \arrow[swap]{dr}{\pi_2} \\
              X & & Y
            \end{tikzcd}
          \]
          commute, then $\stikzfig{Prod2}$ and $\stikzfig{Prod3}$. Therefore we have
          \[
            \stikzfig{Prod4}
          \]
  where the step marked $*$ uses compatibility of the comonoid with the monoidal structure and the fact that $h$ is a comonoid homomorphism.
  This shows that $\alpha$ is unique and therefore $X \otimes Y$ is the product of $X$ and $Y$.
  \end{itemize}
\end{proof}

\section{Cartesian bicategories}
\label{sec:CartBicat}

Going from a functional setting to a relational one, we need to reevaluate our intuition about diagrams. While a function $\stikzfig{f}$
has clearly defined input and output, these notions do not in general make sense for relations. For that reason we will draw arbitrary morphisms as
$\stikzfig{R}$ in the following. The relational analogue of cartesian categories, cartesian bicategories, have axioms very similar
to the former, but with some important differences.

\begin{defi}\label{def:cartBicat}
  A cartesian bicategory is a symmetric monoidal category $(\mathcal{B},\otimes, I)$ enriched over the category of posets.
  Every object $X \in \mathcal{B}$ is equipped with morphisms
  \[
  \stikzfig{CopyX} \from X \to X \otimes X \quad \text{and} \quad \stikzfig{DisX} \from X \to I
  \]
such that
\medskip
  \begin{enumerate}
    \item $\stikzfig{CopyX}$ and $\stikzfig{DisX}$ form a cocommutative comonoid, that is they satisfy
      \begin{align*}
       \sxtikzfig{Assoc} &
  \;\;\;   \;\;\; \;\;\; \;\; \sxtikzfig{Comm} &
       \sxtikzfig{Unit} \\
      \end{align*}
    \item $\stikzfig{CopyX}$ and $\stikzfig{DisX}$ have right-adjoints $\stikzfig{CocopyX}$ and $\stikzfig{CodisX}$ respectively, that is
      \begin{align*}
      & \sxtikzfig{UnitCopy}
      & \sxtikzfig{CounitCopy} \\[5pt]
      & \sxtikzfig{UnitDis}
      & \sxtikzfig{CounitDis}
      \end{align*}
    \item The Frobenius law holds, that is
      \[
        \sxtikzfig{Frob}
      \]
    \item Each morphism $R \from X \to Y$ is a lax comonoid homomorphism, that is
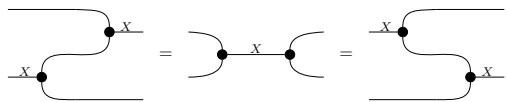
    \begin{align*}
        & \sxtikzfig{LaxCopy}
        & &\sxtikzfig{LaxDis}
    \end{align*}
    \item The choice of comonoid on every object is compatible with the monoidal structure in the sense that
    \begin{align*}
      & \sxtikzfig{CohDis}
      & & \sxtikzfig{CohCopy}
    \end{align*}
    and
    \begin{align*}
        & \sxtikzfig{DisI}
      & & \sxtikzfig{CopyI}
    \end{align*}
  \end{enumerate}
\end{defi}

\begin{defi}
  A morphism of cartesian bicategories is a monoidal functor preserving the ordering and the chosen monoids and comonoids.
\end{defi}

\begin{rem}
  Definition~\ref{def:cartBicat} is a slight deviation from the terminology used in~\cite{carboni1987cartesian}. What we simply call a cartesian bicategory
  here is called a cartesian bicategory of relations in~\cite{carboni1987cartesian}. Moreover, in the original definition of~\cite{carboni1987cartesian}, property (5) is replaced by requiring the uniqueness of the comonoid/monoid. However, as suggested in~\cite{2017arXiv170600526P}, compatibility with the
  monoidal structure seems to be the property of primary interest.
\end{rem}

The archetypal example of a cartesian bicategory is the category of sets and relations $\Rel$, with cartesian product of sets, hereafter denoted by $\times$, as monoidal product and $1=\{\bullet\}$ as unit $I$. To be precise, $\Rel$ has sets as objects and relations $R \subseteq X \times Y$ as arrows $X \to Y$. Composition and monoidal product are defined as expected:
\[R \seq S=\{(x,z) \,| \, \exists y \text{ s.t. } (x,y)\in R\text{ and } (y,z)\in S\},\]
\[R\otimes S =\{\big((x_1,x_2)\,, \, (y_1,y_2)\big) \,|\, (x_1,y_1)\in R \text{ and } (x_2,y_2)\in S \}.\]
 For each set $X$, the comonoid structure is given by the diagonal function $X\to X\times X$ and the unique function $X\to 1$, considered as relations. That is $\stikzfig{CopyX} = \{\big(x, (x,x) \big) \, |\, x\in X\}$ and $\stikzfig{DisX}=\{(x,\bullet) \,|\, x\in X\}$. Their right adjoints are given by their opposite relations: $\stikzfig{CocopyX}=\{\big((x,x),x \big) \, |\, x\in X\}$ and $\stikzfig{CodisX}=\{(\bullet,x) \,|\, x\in X\}$.
 Following the analogy with $\Rel$, we will often call arbitrary morphisms of a cartesian bicategory relations.

There are many examples of cartesian bicategories that are somewhat similar to $\Rel$, for instance $\LinRel$,
the category of linear relations of vector spaces where the monoidal product is the direct sum of vector spaces,
for further details see~\cite{JoshConcur}.
Nevertheless, there are examples of cartesian bicategories that are significantly different, i.e. that are not a form of $\Rel$ with additional structure.
We will show some of those examples at the end of this section (Example \ref{ex:ERPER}), while $\Rel$ will serve to drive our intuition.

\medskip

We commence the exploration of the theory of cartesian bicategories with an elementary fact about the right adjoints.
\begin{lem}
 $\stikzfig{CocopyX}$ and $\stikzfig{CodisX}$ form a commutative monoid, that is
      \begin{align*}
        &\sxtikzfig{CoAssoc} &
        & \sxtikzfig{CoComm} &
        & \sxtikzfig{CoUnit}
      \end{align*}
\end{lem}
\begin{proof}
  A straightforward way of proving these properties is via the uniqueness of adjoints. In that way, they follow directly from the fact that
  $\stikzfig{CopyX}$ and $\stikzfig{DisX}$ form a cocommutative comonoid. To spell out this abstract observation more concretely,
  here we show how it translates into a proof of the last property:
  \[
    \stikzfig{CoUnitProof1}
  \]
  and conversely
  \[
    \stikzfig{CoUnitProof2}
  \]
\end{proof}

One of the fundamental properties of cartesian bicategories that follows from the existence of the monoid and comonoid on every object is that
every local poset $\Hom_{\Bicat}(X,Y)$ allows to take the intersection of relations and has a top element.

\begin{lem}
  Let $\Bicat$ be a cartesian bicategory and $X,Y \in \Bicat$. The poset $\Hom_{\Bicat}(X,Y)$ has a top element given by $\stikzfig{Top}$ and the meet of relations
  $R,S \from X \to Y$ is given by
  \[ \stikzfig{Meet} \]
\end{lem}
\begin{proof}
  \begin{itemize}
    \item For any relation $R \from X \to Y$ we have
      \[\stikzfig{TopProof}\]
    \item Let $R,S \from X \to Y$. Then we have
      \[\stikzfig{MeetLeqR}\]
      and in the same way
      \[\stikzfig{MeetLeqS}\]
      If furthermore $T \leq R$ and $T \leq S$, then we have
      \[\stikzfig{TLeqMeet}\]
      and therefore $\stikzfig{Meet}$ is the meet of $R$ and $S$.
  \end{itemize}
\end{proof}

The existence of meets allows us to characterise inequalities through equalities. It is generally true in a poset with meets that $x \leq y$
if and only if $x \wedge y = x$, where $\wedge$ denotes the meet. This can for example be found in~\cite{lattices}.

\begin{prop}
  Let $\Bicat$ be a cartesian bicategory and $R,S \from X \to Y$ morphisms. We have $R \leq S$ if and only if
  \[
    \stikzfig{Meet=R}
  \]
  \label{prop:IneqEq}
\end{prop}

An immediate consequence of Proposition~\ref{prop:IneqEq} is that for any $R$ we have \[\stikzfig{MeetRR}\]
and therefore in particular $\stikzfig{Special}$. This is known as the special Frobenius law and is a very common companion of
the Frobenius in practical applications~\cite{ZanasiThesis, bonchi2017refinement}.
Another usual companion is the so-called bone equality, given by $\stikzfig{BoneEq}$. This however, does \emph{not} hold in all cartesian bicategories: for instance in $\Rel$, we have $\stikzfig{BoneEqX}$ if and only if the set $X$ is non-empty. If $X$ is empty, the left-hand side of the equality is the empty relation $\{\} \subseteq 1 \times 1 $, while the right-hand-side is the identity relation $\{(\bullet,\bullet)\}\subseteq 1 \times 1$.

\begin{lem}
  A morphism $F \from \Bicat \to \Bicat'$ of cartesian bicategories is faithful, in the sense of reflecting equality between morphisms,
  if and only if it reflects the ordering.
  \label{lem:FaithfulReflectOrder}
\end{lem}
\begin{proof}
  A morphism that reflects the ordering is faithful because the ordering is reflexive and antisymmetric.
  Conversely a faithful morphism reflects the ordering by Proposition~\ref{prop:IneqEq}.
\end{proof}
The Frobenius law (Property (3) in Definition~\ref{def:cartBicat}) gives a compact closed structure -- in other words, it allows us
to bend wires around. The cup of this compact closed structure is $\stikzfig{Cup}$, the cap analogously $\stikzfig{Cap}$ and the Frobenius implies
the snake equations:
\[
  \stikzfig{Snake}
\]
To appreciate the property that every morphism is a lax-comonoid homomorphism (Property (4) in Definition~\ref{def:cartBicat}), it is useful to spell out its meaning in $\Rel$: in the first inequality, the left and the right-hand side are, respectively, the relations \begin{equation}\label{eq:exlax1}\{\big(x,(y,y)\big) \,| \,  (x,y)\in R \} \quad \text{ and } \quad \{\big(x,(y,z)\big) \,| \,  (x,y)\in R \text{ and } (x,z) \in R \}\text{,}\end{equation}
while in the second inequality, they are the relations
\begin{equation}\label{eq:exlax2}\{(x,\bullet) \,| \, \exists y\in Y \text{ s.t. } (x,y)\in R \} \quad \text{ and } \quad  \{(x,\bullet) \,|\, x \in X  \}\text{.}\end{equation}
It is immediate to see that the two left-to-right inclusions hold for any relation $R\subseteq X \times Y$, while the right-to-left inclusions hold exactly when $R$ is a function: a relation which is single valued and total.
This observation justifies the following definition.

\begin{defi}  \label{defn:sv}
  Let $R$ be a morphism in a cartesian bicategory. We call $R$
  \begin{align*}
    \text{single valued if } & \stikzfig{Rsv} \\ \medskip
    \text{total if } & \stikzfig{Rtot} \\
\text{injective if } & \stikzfig{Rinj} \\
    \text{surjective if }&\stikzfig{Rsur} \\
  \end{align*}
\end{defi}
By translating the last two inequalities in $\Rel$, similarly to what we have shown in \eqref{eq:exlax1} and \eqref{eq:exlax2}, the reader can immediately check that these correspond to the usual properties of injectivity and surjectivity for relations. Moreover, since the converses of these inequalities hold in cartesian bicategories, the four inequalities are actually equalities.

We can characterise all the notions of Definition~\ref{defn:sv} equivalently in terms of \emph{opposite morphisms} $\op{R}$ which are defined for
any morphism $R$ as follows:\[\stikzfig{RopDef}\]

\begin{prop}
  Let $R$ be a morphism in a cartesian bicategory.
  \begin{align*}
    \stikzfig{RsvAlt} & \text{iff $R$ is single valued. } \\
    \hspace{-7pt} \stikzfig{RtotAlt} & \text{iff $R$ is total.} \\
    \stikzfig{RinjAlt} & \text{iff $R$ is injective.} \\
    \hspace{-7pt} \stikzfig{RsurAlt} & \text{iff $R$ is surjective.} \\
  \end{align*}
  In particular, $R$ is surjective iff $\op{R}$
 is total and $R$ is injective iff $\op{R}$ is single valued.
  \label{prop:svChar}
\end{prop}

\begin{proof}
  We show the proofs for single valued and total. The proofs for injectivity and surjectivity are analogous. The last statement follows from the others and the fact that $\op{\left( \op{R} \right)} = R$
  \begin{itemize}
    \item Let $R$ be single valued. Then
      \[\stikzfig{RsvAltProof}\]
      Conversely, if $\stikzfig{RsvAlt}$, then by the Frobenius law one gets
      \[\stikzfig{RsvProof1}\]
      and from there
      \[\stikzfig{RsvProof2}\]
    \item Let $R$ be total. Then
      \[\stikzfig{RtotAltProof}\]
      Conversely, if $\stikzfig{RtotAlt}$, then
      \[\stikzfig{RtotProof}\]
  \end{itemize}
\end{proof}

\begin{exa}\label{ex:ERPER}
Recall that a prop is a strict symmetric monoidal category where the objects are the natural numbers and monoidal product on objects is addition.
The prop $\mathbf{ERel}$ of \emph{equivalence relations}~\cite{ZanasiThesis,lmcs:4796,coya2017corelations,fong2016algebra,bruni2001some} (also called the prop of \emph{corelations}) has objects natural numbers, where $n\in\mathbb{N}$ is thought of as the finite set $\{0,\dots, n-1\}$.
A morphism $n \to m$ is an equivalence relation on $n+m$. Composition of an equivalence relation on $n+m$ with one on $m+o$ is given by taking the smallest equivalence relation they generate on $n+m+o$ and restricting it to $n+o$. Monoidal product is given by disjoint union.

Another important example is the prop $\mathbf{PERel}$ of \emph{partial equivalence relations}. These are symmetric and transitive, but not necessarily reflexive, and have been used in the study of the semantics of higher order $\lambda$-calculi~\cite{jacobs1999categorical,streicher2012semantics} and quantum computations~\cite{jacobs2012coreflections,hasuo2017semantics}. In $\mathbf{PERel}$ a morphism $n \to m$ is a partial equivalence relation on $n+m$; composition similar to that in $\mathbf{ERel}$, taking the smallest induced partial equivalence relation. Again $\otimes$ is given by disjoint union. See~\cite[Definitions 2.52 and 2.63]{ZanasiThesis} for additional details.

Both $\mathbf{ERel}$ and $\mathbf{PERel}$ carry the structure of cartesian bicategories after taking into consideration their posetal enrichment. Here the ordering $\leq$ is the \emph{opposite} of set inclusion: $R\leq S$ iff $R\supseteq S$.
Note that for $\mathbf{PERel}$, we need some extra care. We consider partial equivalence relations $R,S\colon n \to m$ as equivalence relations $\bar{R},\bar{S}$ over $(n+m) \cup \{\bot\}$ and then take $R\leq S$ iff $\bar{R} \supseteq \bar{S}$. In particular, notice that the completely undefined partial equivalence relation is represented by the chaotic relation on  $(n+m) \cup \{\bot\}$, and is thus---according to this ordering---the least element in its homset.

To define the comonoid structure it is enough to consider the object $1$, since for arbitrary $n$ it is forced by compatibility with the monoidal structure (Property (5) in Definition \ref{def:cartBicat}). For both  $\mathbf{ERel}$ and $\mathbf{PERel}$ $ \stikzfig{Copy} \colon1 \to 2$ is the equivalence relation equating all the elements of the set $1+2$ and $\stikzfig{Dis} \colon1 \to 0$ equates the single element of the set $1$. The monoid structure $\stikzfig{Cocopy}\colon 2 \to 1$ and $\stikzfig{Codis}\colon 0 \to 1$ is defined in a similar way.
\end{exa}

\section{Maps}
\label{sec:Maps}

In the previous section, we introduced cartesian bicategories and showed some of their fundamental properties.
In this section we focus on a very important class of morphisms, the maps, that behave very much like functions behave in $\Rel$.

\begin{defi}
  A map in a cartesian bicategory is a morphism $f$ that is a comonoid homomorphism, i.e. is single valued and total.
\end{defi}

For maps it makes sense to imagine a flow of information from left to right. We will therefore write $\stikzfig{Map}$ to denote a map $f$
and $\stikzfig{fop}$ for its opposite. We will also refer to $\stikzfig{fop}$ as a comap.
The notation is suggestive of the fact that maps form a Cartesian category as we will see in Lemma~\ref{lem:MapCart}.
Note that we use lower-case letters for maps and upper-case for arbitrary morphisms.

\begin{exa}
 As expected, maps in $\Rel$ are exactly (graphs of) functions: this is easily verified by means of equations \eqref{eq:exlax1} and \eqref{eq:exlax2}.
 In $\LinRel$, maps are exactly linear maps between vector spaces.
\end{exa}

The original treatment of cartesian bicategories in~\cite{carboni1987cartesian} introduces maps as those morphisms
that admit a right-adjoint. We show below that this amounts to the same notion.

\begin{prop} \label{prop:MapRightAdj}
  A morphism $\stikzfig{Map}$ is a map if and only if it has a right adjoint -- a morphism $R$ such that
  $\stikzfig{SVAdj}$ and $\stikzfig{EntireAdj}$.
  In that case, necessarily \[\stikzfig{g=fop}\text{.}\]
\end{prop}
\begin{proof}
  If $\stikzfig{Map}$ is a map, then $\stikzfig{fop}$ is a right-adjoint by Proposition~\ref{prop:svChar}.

  On the other hand, if $\stikzfig{Map}$ has a right-adjoint $R$, then it is a map since
  \[
    \stikzfig{MapComultProof}
  \]
  and
  \[
    \stikzfig{MapDisProof}
  \]
  Therefore, $\stikzfig{Map}$ is indeed a map and $\stikzfig{g=fop}$ by uniqueness of adjoints.
\end{proof}

Since by definition of cartesian bicategories $\stikzfig{Copy}$ has right-adjoint $\stikzfig{Cocopy}$ and $\stikzfig{Dis}$ has right-adjoint
$\stikzfig{Codis}$, we get that they are in fact maps.

\begin{prop}
  Let $\Bicat$ be a cartesian bicategory and $X \in \Bicat$ an object. Then $\stikzfig{CopyX}$ and $\stikzfig{DisX}$ are maps.
\label{prop:ComonoidMap}
\end{prop}

The identity is a map, and maps are easily shown to be closed under composition, so they constitute a category.
\begin{defi}
  Given a cartesian bicategory $\mathcal{B}$, we define its category of maps, $\Map(\mathcal{B})$ to have the same objects of $\mathcal{B}$ and as morphism the maps of $\mathcal{B}$.
  Dually, we define its category of comaps, $\Comap(\Bicat)$ to have the same objects as $\Bicat$ and as morphisms the comaps of $\Bicat$.
\end{defi}

\begin{rem}
  Clearly $\Comap(\Bicat) \cong \op{\Map(\Bicat)}$, by taking a map to its opposite.
\end{rem}

By the following proposition, the ordering of $\mathcal{B}$ becomes trivial when restricted to maps.
\begin{prop}
  Let $f,g$ be maps such that $f \leq g$. Then $f = g$.
  \label{prop:mapCmp}
\end{prop}
\begin{proof}
  Since $\stikzfig{fg}$, also $\stikzfig{fopgop}$. Therefore
  \[
    \stikzfig{gleqf}
  \]
\end{proof}

By analogy with $\Rel$, we think of the comonoid on any object in a cartesian bicategory as giving a way to copy and discard. By definition, maps
respect these operations. Therefore, $\Map(\Bicat)$, which inherits the monoidal product from $\Bicat$ has the structure of a Cartesian category.

\begin{lem}
  For a cartesian bicategory $\Bicat$, the monoidal product $\otimes$ induces a product on $\Map(\Bicat)$. The monoidal unit $I$ becomes a terminal
  object in $\Map(\Bicat)$. In other words, $(\Map(\Bicat), \otimes, I)$ is a cartesian category.
  \label{lem:MapCart}
\end{lem}
\begin{proof}
  Since by Proposition~\ref{prop:ComonoidMap} the comonoid on every object consists of maps, every object in $\Map(\Bicat)$ is equipped with a comonoid
  structure. Every map, by definition, is a comonoid homomorphism respecting this structure, so that $\Map(\Bicat)$ is a cartesian category.
  By Proposition~\ref{prop:CartCopy}, $\otimes$ induces a product and $I$ is a terminal object.
\end{proof}

\begin{exa}
Recall the cartesian bicategories $\mathbf{ERel}$ of equivalence relations and $\mathbf{PERel}$ of partial equivalence relations from Example~\ref{ex:ERPER}.
In order to illustrate what maps are in these categories, it is convenient to write $[i]_R$ for the set $\{j \mid (i,j)\in R\}$. Both in $\mathbf{ERel}$ and $\mathbf{PERel}$ a morphism $R \from n \to m$ is
\begin{equation}\label{eq:totalERel}
\text{ total iff for all }i,j\in n, (i,j)\in R\text{ implies }i=j\text{, and}
\end{equation}
\vspace{-16pt}
\begin{equation}\label{eq:svERel}
\text{single valued iff for all }i\in m,\text{ either }[i]_R=\emptyset\text{ or there is }j\in n\text{ such that }(i,j)\in R.
\end{equation}
Thus $\stikzfig{Cocopy}$ is single valued but not total; $\stikzfig{Codis}$ is total but not single valued.
In $\mathbf{PERel}$, the undefined relation $0\to 1$, hereafter denoted $\stikzfig{Point}$, is both total and single valued.
\end{exa}

\section{Choice in Cartesian bicategories}
\label{sec:AC}
In Section~\ref{sec:CartBicat} we have recalled cartesian bicategories and in Section~\ref{sec:Maps} we have seen  their (cartesian) categories of maps. This raises a natural question: is it possible to reconstruct in some way a cartesian bicategory from its category of map?

In this paper we will face this problem by showing \emph{when} this is possible. It turns out that the answer is closely related to the axiom of choice.

\medskip

One of the many equivalent formulations of the axiom of choice in set theory is
\begin{center}
  Every total relation contains a map.
\end{center}
In a total relation every element in the domain is related to at least one element in the codomain. A map is obtained by choosing, for each element in the domain, exactly one related element in the codomain. This can be stated in the language of cartesian bicategories.

\begin{defi}[Choice]\label{defn:Choice}
  Let $\mathcal{B}$ be a cartesian bicategory. We say that $\mathcal{B}$ satisfies the axiom of choice (AC), or that $\mathcal{B}$ has choice, iff the following holds for any morphism $R \from X \to Y$:
  \begin{equation}\label{eq:AC}
  \stikzfig{Rtot} \text{ ($R$ is total)} \quad \text{implies} \quad \exists \text{ map } f \from X \to Y \text{ such that } \stikzfig{fleqR}  \tag{AC}
  \end{equation}
\end{defi}

Observe that the converse implication holds in any cartesian bicategory.
\begin{lem}
If $\stikzfig{fleqR}$ then $\stikzfig{Rtot}$.
\end{lem}

\begin{proof} Obvious, since if $S$ is total and $S \leq R$, then $R$ is total: $\stikzfig{SEntire}$.
\end{proof}

\begin{exa}\label{ex:AC}~
   \begin{itemize}
   \item
   The usual axiom of choice implies that $\Rel$ satisfies~\eqref{eq:AC}.
   \item
   $\mathbf{ERel}$ is an example of a cartesian bicategory that does not satisfy~\eqref{eq:AC}. Recall from Example~\ref{ex:ERPER} that the ordering is the \emph{reverse} of inclusion. Therefore, for~\eqref{eq:AC} to hold would mean that every equivalence relation that satisfies~\eqref{eq:totalERel} could be included in one that satisfies both~\eqref{eq:totalERel} and~\eqref{eq:svERel}. Now consider $\stikzfig{Codis}\colon 0 \to 1$. As seen in Example \ref{ex:ERPER}, it is total, but not single valued. Since equivalence relations have to be reflexive, this is also the only morphism of type $0 \to 1$: clearly AC fails here.
   \item Interestingly, $\mathbf{PERel}$ \emph{does} satisfy~\eqref{eq:AC}. For example, $\stikzfig{Codis}\colon 0 \to 1$ is included, as an
   equivalence relation over $(0+1)\cup\{\bot\}$, in $\stikzfig{Point}$.
  \end{itemize}
\end{exa}

Another common formulation of the axiom of choice in set theory is the assertion that every surjective function $\pi \colon X \to Y$ splits, namely, there exists a function $\rho\colon Y \to X$ such that $\rho \seq \pi = id_Y$.
A standard categorification of the notion of surjectivity is the notion of epi(morphism): $\pi$ is epi iff $\pi \seq f=\pi \seq g$ entails $f=g$.
In order to clarify the picture and justify our Definition~\ref{defn:Choice}
we will now investigate epimorphisms in cartesian bicategories.

\begin{lem}
  Let $\stikzfig{MapPi}$ be a map in a cartesian bicategory $\mathcal{B}$. Then $\stikzfig{MapPi}$ is an epi in $\mathcal{B}$ if and only if it is surjective.
  \label{lem:EquivSurjective}
\end{lem}
\begin{proof}
  \begin{itemize}
    \item Let $\pi$ be an epi in $\mathcal{B}$. Since $\pi$ is a map, by Proposition \ref{prop:svChar}, $\stikzfig{PiPiopEntire}$ and therefore
      \[\stikzfig{PiPreUnit}\]
      Since $\stikzfig{MapPi}$ is epi, $\stikzfig{PiUnit}$ so $\stikzfig{ComapPi}$ is total, hence $\stikzfig{MapPi}$ is surjective by Proposition~\ref{prop:svChar}. \smallskip
    \item Assume $\stikzfig{MapPi}$ is surjective. Then $\stikzfig{PiAdj}$ by Proposition~\ref{prop:svChar}. \smallskip
      If now $R,S$ are morphisms such that $\stikzfig{pif=pig}$, then \[\stikzfig{f=g}\]
  \end{itemize}
\end{proof}

\begin{lem}
  Surjective maps split in any cartesian bicategory with choice.
  \label{lem:ChoiceSplitsEpis}
\end{lem}
\begin{proof}
  Let $\pi \from X \to Y$ be a surjective map.
  Therefore, $\op{\pi} \from Y \to X$ is a total relation, so by (AC) there is a map $g \from Y \to X$ such that
  \[
    \stikzfig{EpiSplit1}
  \]
  Now we have
  \[
    \stikzfig{EpiSplit2}
  \]
  and since both the left hand side and the right hand side of that inequality are maps, we have by Proposition~\ref{prop:mapCmp} that $g \seq \pi = \id_{Y}$.
\end{proof}

\subsection{Cartesian bicategories with enough maps}

The converse of Lemma~\ref{lem:ChoiceSplitsEpis} does not hold in general.
The reason is that a general cartesian bicategory might not have enough maps to ``cover'' all its morphisms in a suitable sense.
In order to prove the converse, we need to assume a saturation property.

\begin{defi}
  We say a cartesian bicategory \emph{has enough maps} if for every morphism $R \from X \to I$
  there is a map $f \from Z \to X$ such that
  \[
    \stikzfig{EnoughMaps}
  \]
  \label{defn:enoughMaps}
\end{defi}

The intuition for this notion is the following: a morphism $R \from X \to I$ can be considered as a predicate on $X$. Then having enough maps ensures the existence of a function $f$ that picks out the subset of $X$ where $R$ holds.

\begin{exa}\label{ex:enoughmap} The description above shows that $\Rel$ has enough maps.
Also $\mathbf{ERel}$ and $\mathbf{PERel}$ have both enough maps. We briefly describe the construction for $\mathbf{ERel}$, the one for $\mathbf{PERel}$ is similar. For any morphism $R\colon n \to 0$ in $\mathbf{ERel}$, take $e$ to be the number of the equivalence classes of $R$. Choose a total ordering for these equivalence classes, so that for each $i\in e=\{0,\dots e-1\}$, we denote by $R_i$ the $i$-th equivalence class of $R$.
Then, define $f\colon e \to n$ as the equivalence on $e+n$
\[R \cup \{(i,j) \, |\, i\in e \text{ and } j\in R_i  \} \cup  \{(i,j) \, |\, j\in e \text{ and } i\in R_j  \}\text{.}\]
It is immediate to see that $f$ satisfies \eqref{eq:totalERel} and \eqref{eq:svERel} and that $ \stikzfig{EnoughMaps}$.
\end{exa}

\begin{rem}
  A similar property, \emph{functional completeness}, was already considered in \cite{carboni1987cartesian}.
  The important difference is that we don't require $f$ to be mono.
  Ours is a more general notion: every functionally complete cartesian bicategory
  also has enough maps.
  \label{rem:FuncComp}
\end{rem}

\begin{lem}\label{RFact}
  If a cartesian bicategory has enough maps, then for every morphism $R \from X \to Y$,
  there are maps $f \from Z \to X$ and $g \from Z \to Y$ such that
  \[
    \stikzfig{RFact}
  \]
  We call this a \emph{comap-map factorisation} of $R$.
\end{lem}
\begin{proof}
  Since there are enough maps, there is a map $h \from Z \to X \otimes Y$ such that
  \[
    \stikzfig{RFact1}
  \]
  Let $\stikzfig{RFact2}$ and $\stikzfig{RFact3}$, then
  \[
    \stikzfig{RFact4}
  \]
  where the step marked $*$ uses compatibility of the comonoid with the monoidal structure and the fact that $h$ is a map.
  Therefore we have
  \[
    \stikzfig{RFact5}
  \]
\end{proof}

\begin{prop}
  A cartesian bicategory with enough maps satisfies~\eqref{eq:AC} iff surjective maps split.
  \label{prop:ACiffEpisSplit}
\end{prop}
\begin{proof}
  By Lemma~\ref{lem:ChoiceSplitsEpis}, it suffices to prove that~\eqref{eq:AC} holds if surjective maps split.
  So let $R \from X \to Y$ be a total relation and take a comap-map factorisation $\stikzfig{RFact}$ with maps $f,g$.
  Since $R$ is total,
  \[
    \stikzfig{EpiSplit3}
  \]
  so $f$ is surjective. Since surjective maps split,
  there exists a map $h$ that is a pre-inverse of $f$, so $h \seq f = \id$. Then
  \[
    \stikzfig{EpiSplit4}
  \]
  and therefore $\stikzfig{hgR}$, so $R$ contains a map.
\end{proof}

\medskip

We can now provide a first answer to the question posed at the beginning of this section: all cartesian bicategories with enough maps that satisfy (AC) can be recovered from their category of maps.

\begin{thm}
  Let $\Bicat$ be cartesian bicategory with enough maps. Then $\Bicat$ satisfies the axiom of choice iff
    \[
      \Bicat \cong \Spantilde{\Map(\Bicat)}
    \]
  \label{thm:Choice}
\end{thm}

Here, $\mathsf{Span}^{\sim}$ is the key construction that we will introduce in the next section. Observe that the above theorem not only states that this reconstruction is feasible, but it also characterises cartesian bicategories with choice amongst those with enough maps. The proof of this result will exploit a more general construction, called $\mathsf{Span}^{\mathcal{S}}$, that we will introduce in Section \ref{sec:spanS}.

\section{The \texorpdfstring{$\mathsf{Span}^{\sim}$}{Spantilde} construction}
\label{sec:Span}

Our starting observation is that in an arbitrary cartesian bicategory, commutative diagrams of maps give rise to inequalities in a very straightforward manner.

\begin{lem}
  Let $\mathcal{B}$ be a cartesian bicategory and
  \[
  \begin{tikzcd}[sep = tiny]
    & \arrow[swap]{dl}{f} A \arrow{dd}{\alpha} \arrow{dr}{g} & \\
    B & & C \\
    & \arrow{ul}{h} D \arrow[swap]{ur}{k} & \\
  \end{tikzcd}
  \]
  a commutative diagram of maps. Then $\stikzfig{Witness}$.
  \label{lem:Filler}
\end{lem}
\begin{proof}
  \[
    \stikzfig{CommDiagMapProof}
  \]
\end{proof}

As a matter of fact, in a cartesian bicategory that satisfies~\eqref{eq:AC}, all inequalities are of this form, so we have
the following result, which is a converse of Lemma~\ref{lem:Filler}.

\begin{lem}
  Let $\mathcal{B}$ be a cartesian bicategory with choice and
  
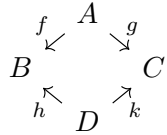
\[
  \begin{tikzcd}[sep = tiny]
    & \arrow[swap]{dl}{f} A \arrow{dr}{g} & \\
    B & & C \\
    & \arrow{ul}{h} D \arrow[swap]{ur}{k} & \\
  \end{tikzcd}
  \]
  a diagram of maps such that
  $  \stikzfig{Witness}$.
  Then there is a map $\omega \from A \to D$ such that the following diagram commutes.
  
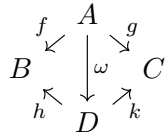
\[
  \begin{tikzcd}[sep = tiny]
    & \arrow[swap]{dl}{f} A \arrow{dd}{\omega} \arrow{dr}{g} & \\
    B & & C \\
    & \arrow{ul}{h} D \arrow[swap]{ur}{k} & \\
  \end{tikzcd}
  \]
  \label{lem:FillSquare}
\end{lem}
\begin{proof}
  Consider $R \from A \to D$ given by
  \[
    \stikzfig{Witness2}
  \]
  One readily checks that $\stikzfig{WitnessProp1}$ and $\stikzfig{WitnessProp2}$.

  $R$ is total, since
  \[
    \stikzfig{Witness3}
  \]
  so by the axiom of choice, there is a map $\omega \leq R$.
  This satisfies
  \[
    \stikzfig{Witness4}
  \]
  and
  \[
    \stikzfig{Witness5}
  \]
  and since both side are maps we have equality by Proposition \ref{prop:mapCmp}.
\end{proof}

Lemma~\ref{lem:FillSquare} gives us a characterisation of inequalities in the presence of~\eqref{eq:AC}.
This motivates us to synthesise cartesian bicategories where the ordering is directly defined in this way.

\begin{defi}[$\mathsf{Span}$]
  \label{def:Span}
Let $\Cat$ be a finitely complete category. A \emph{span} from $X$ to $Y$ is a pair of arrows $X \leftarrow A \to Y$ in $\Cat$.

A morphism $\alpha \from (X \leftarrow  A \to Y) \Rightarrow (X \leftarrow  B \to  Y)$ is an arrow $\alpha \from A \to B$ in $\Cat$ s.t. the
following diagram commutes:
 
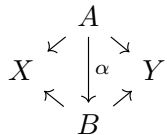
\[
   \begin{tikzcd}[sep=tiny]
      & \arrow{dl} \arrow{dr} A \arrow{dd}{\alpha} & \\
      X & & Y \\
      & \arrow{ul} \arrow{ur} B & \\
   \end{tikzcd}
 \]
Two spans $X \leftarrow A \to Y$ and $X \leftarrow B \to Y$ are \emph{isomorphic} if $\alpha$ is an isomorphism.
For $X \in \Cat$, the \emph{identity span} is $X \xleftarrow{\id_X} X \xrightarrow{\id_X} X$. The composition of spans $X\leftarrow  A \xrightarrow{f} Y$
and $Y\xleftarrow{g}  B \to Z$ is $X \leftarrow A \times_{f,g} B \to Z$,
obtained by taking the pullback $A \times_{f,g} B$ of $f$ and $g$. This data defines the bicategory~\cite{benabou1967introduction} $\Span{\Cat}$:
the objects are those of $\Cat$, the arrows are spans and 2-cells are homomorphisms.
Finally, $\Span{\Cat}$ has monoidal product given by the product in $\Cat$, with unit the final object $1\in \Cat$.
\end{defi}

To avoid the complications that come with bicategories, such as composition being associative only up to isomorphism,
it is common to consider a \emph{category} of spans, where isomorphic spans are equated: let $\Spanleq{\Cat}$ be the monoidal category that has
isomorphism classes of cospans as arrows. Note that, when going from bicategory to category,
after identifying isomorphic arrows it is usual to simply discard the 2-cells.
Differently, we consider $\Spanleq{\Cat}$ to be locally preordered with $(X \leftarrow A \to Y) \leq (X \leftarrow B \to Y)$ if there exists a morphism
$\alpha \from (X \leftarrow A \to Y) \Rightarrow (X \leftarrow B \to Y) $. It is an easy exercise to verify that this (pre)ordering is well-defined and compatible with composition and monoidal product. Note that, in general, $\leq$ is a genuine preorder: i.e. it is possible that $(X \to A \leftarrow Y) \leq (X \to B \leftarrow Y) \leq (X \to A \leftarrow Y)$ without the cospans being isomorphic.

Since $\Spanleq{\Cat}$ is preorder enriched, rather than poset enriched, it is \emph{not} a cartesian bicategory.
However, one can transform a preorder enriched category into a poset enriched one with a simple construction:
for $\Spanleq{\Cat}$, one first defines $\sim=\leq \cap \geq$, namely
$(X \leftarrow A \to Y) \sim (X \leftarrow B \to Y)$ iff there exists
$\alpha \colon (X \leftarrow A \to Y) \Rightarrow (X \leftarrow B \to Y)$ and
$\beta \colon (X \leftarrow B \to Y) \Rightarrow (X \leftarrow A \to Y)$,
and then one takes equivalence classes of morphisms of $\Spanleq{\Cat}$ modulo $\sim$. It is worth observing that pullbacks are no longer necessary
to compose $\sim$-equivalence classes of spans: weak pullbacks are sufficient,
since non-isomorphic weak pullbacks of the same cospan all belong to the same $\sim$-equivalence class.
We therefore define the posetal category $\Spantilde{\Cat}$ which has the same objects as $\Cat$ and as morphisms $\sim$-equivalence classes of spans.
The order is defined as in $\Spanleq{\Cat}$.
Composition is given by weak pullbacks in $\Cat$. Identities, monoidal product and unit are as in $\Span{\Cat}$.

The construction of $\Spantilde{\Cat}$ has appeared in~\cite{guitart1980relations} under the name of $\mathsf{REL}(\Cat)$.

\begin{prop}
Let $\Cat$ be a category with finite products and weak pullbacks. Then $\Spantilde{\Cat}$ is a cartesian bicategory.
\end{prop}
\begin{proof}

For $\stikzfig{CocopyX}$ we take the span $X\times X \leftarrow X \to X $ and for $\stikzfig{CodisX}$ we take $1 \leftarrow X \to X$.

With this information, one has only to check that the inequalities in Definition \ref{def:cartBicat} hold: each of them is witnessed  by a commutative diagrams in $\mathcal{C}$. As an example, we illustrate $\stikzfig{UnitDis}$.
    The left hand side is the span  $X \xleftarrow{\id_X} X \xrightarrow{\id_X} X$.
    The right hand side is the composition of $X \xleftarrow{\id_X} X \xrightarrow{!} 1$ and $1 \xleftarrow{!} X \xrightarrow{\id_X} X$. Since the product $X \xleftarrow{\pi_1} X \times X \xrightarrow{\pi_2} X$ is a pullback of $X \xrightarrow{!}1 \xleftarrow{!}X$,
    the composition turns out to be exactly the span $X \xleftarrow{\pi_1} X \times X \xrightarrow{\pi_2} X$.
    Now the diagonal $\Delta \from X \to X \times X$ makes the following diagram in $\mathcal{C}$ commute. Therefore $\Delta$ witnesses the inequality $\stikzfig{UnitDis}$.
    \[
      \begin{tikzcd}[sep=tiny]
        & \arrow[swap]{dl}{\id_X} X \arrow{dd}{\Delta} \arrow{dr}{\id_X} & \\
        X & & X \\
        & \arrow{ul}{\pi_1} X \times X \arrow[swap]{ur}{\pi_2} &
      \end{tikzcd}
    \]
\end{proof}

\begin{prop}
  Let $\Cat$ be a category with finite products and weak pullbacks. Then $\Map(\Spantilde{\Cat}) \cong \Cat$ and surjective maps in $\Spantilde{\Cat}$ are exactly split epis in $\Cat$.
  \label{prop:MapSpan}
\end{prop}
\begin{proof}
Since $\Cat$ has finite products, it is endowed with a cartesian monoidal structure. This means in particular that  $\stikzfig{mapSpan2}$ for all $g$ in $\Cat$.

  Let $F \from \Cat \to \Spantilde{\Cat}$ be the identity on objects and mapping a morphism $f \from X \to Y$ to the span
  $X \xleftarrow{\id_X} X \xrightarrow{f} Y$. It is easy to check that $F$ is a monoidal functor.

  \noindent
  Since every morphism in $\Cat$ is a comonoid homomorphism,
  $F$ factors as $\Cat \xrightarrow{F'} \Map(\Spantilde{\Cat}) \rightarrow \Spantilde{\Cat}$.
To conclude that $F'$ is an isomorphism, it is enough to show that every $\Spantilde{\Cat}$ map is the $\sim$-equivalence class of some span $X \xleftarrow{\id_X} X \xrightarrow{f} Y$.

  Now, if $X \xleftarrow{f} Z \xrightarrow{g} Y$ is a map in $\Spantilde{\Cat}$, in particular
  \[\stikzfig{MapSpan}\]
  Therefore, by Definition of the ordering in $\Spantilde{\Cat}$, there is a morphism $h \from X \to Z$ such that
  \[
    \begin{tikzcd}[sep=tiny]
      & X \arrow[swap]{dl}{\id_X} \arrow{dd}{h} \\
      X \\
      & Z \arrow{ul}{f}
    \end{tikzcd}
  \quad
  \text{commutes, and therefore}
  \quad
    \begin{tikzcd}[sep=tiny]
      & X \arrow[swap]{dl}{\id_X} \arrow{dd}{h} \arrow{dr}{h \seq g} \\
      X & & Y \\
      & Z \arrow{ul}{f} \arrow[swap]{ur}{g}
    \end{tikzcd}.
  \]
  The two spans are thus equal in $\Spantilde{\Cat}$, since they are both maps.

\medskip

We can now prove the second part of the proposition.
  If $\pi \from X \to Y$ is a map in $\Cat$ such that $F(\pi)$ is surjective in $\Spantilde{\Cat}$, then we have
  \[\stikzfig{EpiSpan}
  \quad
  \text{and therefore there is $\iota \from Y \to X$ such that}
  \quad
    \begin{tikzcd}[sep=tiny]
      & Y \arrow{dr}{\id_Y} \arrow{dd}{\iota} \\
      & & Y \\
      & X \arrow[swap]{ur}{\pi}
    \end{tikzcd}
  \]
  so $\pi$ is a split epi. The converse direction that split epis are surjective maps is obvious.
\end{proof}

\begin{prop}
$\Spantilde{\mathcal{C}}$ has enough maps.
\end{prop}
\begin{proof}
In a cartesian bicategory, for all $R\colon X \to I$ we have $R\leq  \stikzfig{DisX}$. In the special case when $R$ is a map $g\colon X \to I$, by Proposition~\ref{prop:mapCmp}, it holds that $g=\stikzfig{DisX}$. Now take a morphism $R\colon X \to I$ in $\Spantilde{\mathcal{C}}$. By definition, $R$ is a span $X \xleftarrow{f} A \xrightarrow{g} I$. Observe that by Proposition~\ref{prop:MapSpan}, both $f$ and $g$ are maps in $\Spantilde{\mathcal{C}}$. Therefore $g=\stikzfig{DisA}$ and $\Spantilde{\mathcal{C}}$ has enough maps.
\end{proof}

By Proposition \ref{prop:ACiffEpisSplit}, the two propositions above entail the following.

\begin{cor}\label{cor:spanAC}
$\Spantilde{\Cat}$  satisfies~\eqref{eq:AC}.
\end{cor}

The dual of spans, called cospans, are similarly useful for us, so we can make similar definitions in that case.

\begin{defi}
  Let $\Cat$ be a category with finite coproducts and weak pushouts. Then we define
  \[
    \Cospantilde{\Cat} = \Spantilde{\op{\Cat}}
  \]
  Explicitly, $\Cospantilde{\Cat}$ has the same objects as $\Cat$ and a morphism $X \to Y$ is a cospan $X \rightarrow A \leftarrow Y$.
  We have $X \rightarrow A \leftarrow Y \leq X \rightarrow B \leftarrow Y$ if and only if there is a morphism \emph{in the opposite direction}, i.e.
  a morphism $\alpha \from B \to A$ such that
      \[
        \begin{tikzcd}[sep=small]
          & B \arrow{dd}{\alpha} & \\
          X \arrow{dr} \arrow{ur} & & Y \arrow{ul} \arrow{dl} \\
          & A & \\
        \end{tikzcd}
      \]
  \label{def:Cospantilde}
\end{defi}

\begin{exa}\label{ex:Frob}
Let $\FinSet$ be the category with natural numbers as objects and as morphisms functions (as in Example \ref{ex:ERPER}, natural numbers are regarded as finite sets). The category $\Spantilde{\op{\FinSet}} = \Cospantilde{\FinSet}$ satisfies~\eqref{eq:AC} by Corollary~\ref{cor:spanAC}.
This category is particularly relevant for different reasons. First, it is the cartesian bicategory on one object (see~\cite[Theorem 31]{GCQ}) or, using the terminology in \cite{relationaltheories}, it is the Carboni-Walters category freely generated by the empty Frobenius theory. Moreover, after forgetting its posetal enrichment, it is the PROP $\mathbf{Frob}$ of special Frobenius bimonoids which appears to be of fundamental importance in several works (e.g. in \cite{lack2004composing,bonchi2016rewriting}). Finally, the cartesian bicategory of equivalence relations, $\mathbf{ERel}$ from Example \ref{ex:ERPER}, can be obtained as a quotient of $\Cospantilde{\FinSet}$: to pass from cospans to equivalence relations, it suffices to equate $$\stikzfig{BoneEq}\text{.}$$

Since $\mathbf{ERel}$ does not satisfy~\eqref{eq:AC}, by Corollary \ref{cor:spanAC}, there is no category $\Cat$, such that $\mathbf{ERel}$ is $\Spantilde{\Cat}$.
Instead, $\mathbf{PERel}$ can be put in $\mathsf{Span}^{\sim}$ form: it is $\Spantilde{\op{\FinSet_p}} = \Cospantilde{\FinSet_p}$ for $\FinSet_p$ being defined as $\FinSet$ but with partial functions as morphisms. Indeed, as anticipated by Theorem \ref{thm:Choice}, any cartesian bicategory with enough maps that satisfies~\eqref{eq:AC} arises from the $\mathsf{Span}^{\sim}$ construction. In Section \ref{sec:spanS}, we will provide a proof of Theorem~\ref{thm:Choice} by exploiting the more general $\mathsf{Span}^{\mathcal{S}}$ construction. We will see in Example~\ref{ex:ERelSpan}, that  $\mathbf{ERel}$ can be put in $\mathsf{Span}^{\mathcal{S}}$ form.
\end{exa}

\section{Tame Cartesian bicategories}
\label{sec:CartBicatAndMaps}
Theorem~\ref{thm:Choice} provides sufficient conditions for a cartesian bicategory to be reconstructed from its maps by means of  $\mathsf{Span}^{\sim}$. In this section, we weaken those conditions, more precisely (AC), so that a cartesian bicategory can be reconstructed by means of the more general construction $\mathsf{Span}^{\mathcal{S}}$ that we will introduce in Section \ref{sec:spanS}.

We start with a simple observation.
\begin{lem}
Let $\mathcal{B}$ be a cartesian bicategory and consider the following diagram in $\Map(\mathcal{B})$.
  \begin{equation}\label{eq:square}
  \begin{tikzcd}[sep = tiny]
    & \arrow[swap]{dl}{f} A \arrow{dr}{g} & \\
    B \arrow[swap]{dr}{h} & & \arrow{dl}{k} C \\
    & D & \\
  \end{tikzcd}
  \end{equation}
  Then $\stikzfig{CommDiag} $ if and only if the diagram commutes.
  \label{lem:CommMaps}
\end{lem}
\begin{proof}
  If the diagram commutes, then
      \[
        \stikzfig{CommDiag2}
      \]
     Conversely, if $\stikzfig{CommDiag}$, then
      \[
        \stikzfig{CommDiag3}
      \]
      and since both sides are maps, they are equal by Proposition \ref{prop:mapCmp}.
\end{proof}

That means, the category of maps can ``see'' inequalities like the one in Lemma~\ref{lem:CommMaps} as commutative squares.
We are interested in those cartesian bicategories where maps can furthermore identify equalities. How does $\Rel$ do it?

\begin{prop}
  Let \eqref{eq:square} be a commutative diagram in $\Set$. Then $\stikzfig{CommEq}$ in $\Rel$ if and only if \eqref{eq:square} is a weak pullback in $\Set$.
\end{prop}
\begin{proof}
  As Lemma~\ref{lem:CommMaps} shows, $\stikzfig{CommDiag}$ comes from the commutativity of the diagram. Let's unpack the reverse inequality in $\Rel$:
  If $\stikzfig{WeakPB}$, this means that whenever $h(b) = k(c)$, there is an $a \in A$ such that $b = f(a)$ and $c = g(a)$.
  If we let $P = \{(b,c) \mid h(b) = k(c)\}$ with the evident projections $P \to B$ and $P \to C$,
  it is well known that
  \[
    \begin{tikzcd}[sep = tiny]
      & \arrow[swap]{dl} P \arrow{dr} & \\
      B \arrow[swap]{dr}{h} & & \arrow{dl}{k} C \\
      & D & \\
    \end{tikzcd}
  \]
  is a pullback diagram in $\Set$. The maps $f$ and $g$ induce a unique map $\pi \from A \to P$ given by $\pi(a) = (f(a), g(a))$.
  Therefore, by the above discussion, we have $\stikzfig{WeakPB}$ if and only if $\pi$ is surjective.

  It remains to see that $\pi$ is surjective if and only if $A$ is a weak pullback. We will check both implications.
  \begin{itemize}
    \item If $\pi$ is surjective, by the axiom of choice it is a split epi, so there is a map $h \from P \to A$ with $h \seq \pi = \id_{P}$.
        It is straightforward to see that this gives $A$ the universal property of a weak pullback.
    \item Conversely, if $A$ is a weak pullback, there is an induced morphism $h \from P \to A$ and the composite $h \seq \pi$ is compatible with
        the projections $P \to B$ and $P \to C$. Therefore, by the uniqueness clause of the universal property of $P$, we have $h \seq \pi = \id_{P}$,
        so $\pi$ is a split epi and therefore surjective.
  \end{itemize}
\end{proof}

We will now restrict our attention to those cartesian bicategories that have an interplay between relations and maps that is similar to that of $\Rel$.

\begin{defi}
Call a cartesian bicategory $\Bicat$ \emph{tame} if it satisfies the following two conditions:
  \begin{itemize}
    \item $\Bicat$ has enough maps.
    \item $\stikzfig{WeakPB}$ if and only if diagram \eqref{eq:square} is a weak pullback in $\Map(\Bicat)$.
  \end{itemize}
\end{defi}

\begin{rem}
  As proven in~\cite{carboni1987cartesian}, a functionally complete cartesian bicategory allows for pullbacks of maps that satisfy $\stikzfig{WeakPB}$.
  It therefore follows that every functionally complete cartesian bicategory is tame.
\end{rem}

\begin{lem}
  Let $\Bicat$ be a cartesian bicategory with enough maps and choice. Then $\Bicat$ is tame.
  \label{lem:ChoiceEntailsTame}
\end{lem}
\begin{proof}
Consider diagram \eqref{eq:square} in $\Map{(\mathcal{B})}$. We need to prove that $\stikzfig{WeakPB}$ if and only if \eqref{eq:square} is a weak pullback.
  \begin{itemize}
    \item Assume $\stikzfig{WeakPB}$, we want to show that \eqref{eq:square} is a weak pullback.
      Given a commutative diagram of solid arrows below,
      
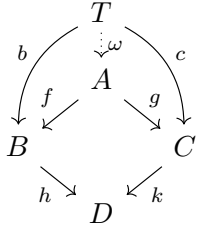
\[
        \begin{tikzcd}[sep=small]
        & T \arrow[dotted]{d}{\omega} \arrow[swap, bend right]{ddl}{b} \arrow[bend left]{ddr}{c} & \\
        & A \arrow[swap]{dl}{f} \arrow{dr}{g} & \\
        B \arrow[swap]{dr}{h} & & \arrow{dl}{k} C \\
        & D &
      \end{tikzcd}
    \]
    we need to construct the dotted arrow. By Lemma~\ref{lem:CommMaps}, , we get
    \[
      \stikzfig{PBProof}
    \]
    and therefore by Lemma~\ref{lem:FillSquare} we get $\omega \from T \to A$ as desired.
  \item Let now~\eqref{eq:square} be a weak pullback diagram. We want to prove that $\stikzfig{WeakPB}$.
    By Lemma~\ref{RFact}, take $\stikzfig{hkFact}$ to be a comap-map factorisation with $\beta \from T \to B$ and $\gamma \from T \to C$.
    By Lemma~\ref{lem:CommMaps}, the external square of the following diagram commutes.
      
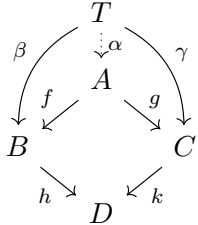
\[
        \begin{tikzcd}[sep=small]
        & T \arrow[dotted]{d}{\alpha} \arrow[swap, bend right]{ddl}{\beta} \arrow[bend left]{ddr}{\gamma} & \\
        & A \arrow[swap]{dl}{f} \arrow{dr}{g} & \\
        B \arrow[swap]{dr}{h} & & \arrow{dl}{k} C \\
        & D &
      \end{tikzcd}
    \]
   Since \eqref{eq:square} is a weak pullback,
 there is $\alpha \from T \to A$ making the above commute.
    With this we get
    \[
      \stikzfig{hkProof}
    \]
  \end{itemize}
\end{proof}

\begin{exa} In Examples \ref{ex:AC} and \ref{ex:enoughmap} we have seen that both $\Rel$ and $\mathbf{PERel}$ have enough maps and choice. Therefore, by Lemma~\ref{lem:ChoiceEntailsTame} both $\Rel$ and $\mathbf{PERel}$ are tame. We will see later in Example~\ref{ex:ERelSpan}, that also $\mathbf{ERel}$ is tame.
\end{exa}

\begin{prop}
  For a tame cartesian bicategory $\Bicat$, there is a morphism of cartesian bicategories $\MorphTame \from \Spantilde{\Map(\Bicat)} \to \Bicat$.
  This $\MorphTame$ is identity on objects and full.
  \label{prop:MorphTame}
\end{prop}
\begin{proof}
  Let $\MorphTame(X \xleftarrow{f} A \xrightarrow{g} Y) = \stikzfig{fopg}$. Since $\Bicat$ is tame, $\Map(\Bicat)$ has well-behaved weak pullbacks
  and therefore $\MorphTame$ preserves composition. It is furthermore monoidal because products in $\Map(\Bicat)$ are induced by the monoidal product
  in $\Bicat$ and it preserves the ordering by Lemma~\ref{lem:Filler}.

\end{proof}

The only thing that prevents $\MorphTame \from \Spantilde{\Map(\Bicat)} \to \Bicat$ from being an isomorphism is that it is in general not faithful.
The reason for that is that the ordering $\Bicat$ might consist of more than just the inequalities that are mediated by morphisms as in Lemma~\ref{lem:Filler}.
We therefore need to characterise the inequalities in $\Bicat$ that we are currently overlooking.

\begin{lem}
  Let $\Bicat$ be a tame cartesian bicategory. Then $\stikzfig{ImageLeq}$ if and only if there is a commutative diagram in $\Map(\Bicat)$
  \[
    \begin{tikzcd}
      & A \arrow[swap]{dl}{f} \\
      X & P \arrow[swap]{u}{\pi} \arrow{d}{\alpha} \arrow{l} \\
      & B \arrow{ul}{g}
    \end{tikzcd}
  \]
  with $\pi$ surjective.
  \label{lem:ImageSurj}
\end{lem}
\begin{proof}
  Assume we have $\alpha$ and $\pi$ such that $\pi \seq f = \alpha \seq g$ with $\pi$ surjective. Then
  \[
    \stikzfig{Image1}
  \]

  Conversely, if $\stikzfig{ImageLeq}$, let
  \[
    \begin{tikzcd}[sep=tiny]
      & \arrow[swap]{dl}{\alpha} P \arrow{dr}{\pi} & \\
      A \arrow[swap]{dr}{g} & & B \arrow{dl}{f} \\
      & X &
    \end{tikzcd}
  \]
  be a weak pullback. This makes the required diagram commute and $\pi$ is surjective, because
  \[
    \stikzfig{Image2}
  \]
\end{proof}

\begin{cor}
  In a tame cartesian bicategory we have
  \[
    \stikzfig{Spanleq}
  \]
  if and only if there is a commutative diagram of maps
  \[
  \begin{tikzcd}[sep = small]
    & \arrow[swap]{dl}{f} A \arrow{dr}{g} & \\
    X & A' \arrow[swap]{u}{\pi} \arrow{d}{\alpha} &  Y\\
    & \arrow{ul}{h} B \arrow[swap]{ur}{k} & \\
  \end{tikzcd}
  \]
  with $\pi$ surjective.
  \label{cor:IneqSpan}
\end{cor}

Therefore, a tame cartesian bicategory is uniquely determined by its category of maps and the knowledge of which maps are surjective.
The class of surjective maps has a number of interesting properties:

\begin{prop}
  Let $\Bicat$ be a tame cartesian bicategory. Let $\mc{S}$ be the class of surjective maps in $\Bicat$.
  \begin{itemize}
    \item $\mc{S}$ contains identities.
    \item $\mc{S}$ is closed under composition.
    \item $\mc{S}$ is closed under products.
    \item $\mc{S}$ is closed under weak pullback.
    \item If $f \seq \pi \in \mc{S}$, then $\pi \in \mc{S}$.
  \end{itemize}
\end{prop}\label{prop:surjectivetame}
\begin{proof}
  It is very straightforward to see that identities are surjective and that $\mc{S}$ is closed under composition and products.
  We will therefore prove the other two properties:
  \begin{itemize}
    \item Let
      \[
        \begin{tikzcd}[sep = tiny]
          & \arrow[swap]{dl}{f} A \arrow{dr}{\tau} & \\
          B \arrow[swap]{dr}{\pi} & & \arrow{dl}{k} C \\
          & D & \\
        \end{tikzcd}
      \]
      be a weak pullback diagram in $\Map(\Bicat)$ with $\pi$ surjective. We want to show that also $\tau$ is surjective.
      Since $\Bicat$ is tame, we have $\stikzfig{WeakPBpi}$.
      Therefore
      \[
        \stikzfig{tauSurj}
      \]
      which means that $\tau$ is surjective.
    \item Let now $f \seq \pi$ be surjective. We want to show that already $\pi$ is surjective. We have
      \[
        \stikzfig{piSurj}
      \]
      and therefore already $\pi$ is surjective.
  \end{itemize}
\end{proof}

\section{The \texorpdfstring{$\mathsf{Span}^{\mathcal{S}}$}{SpanS} construction}\label{sec:spanS}
In this section, we generalise the $\mathsf{Span}^{\sim}$ construction by taking inspiration from the characterisation of the ordering in a tame cartesian bicategory provided by Corollary~\ref{cor:IneqSpan}. We commence by generalising the properties of the class of surjective maps in a tame cartesian bicategory provided by
Proposition \ref{prop:surjectivetame}.

\begin{defi}
  Let $\Cat$ be a category with finite products and weak pullbacks.
  A class of morphisms $\mc{S}$ in $\Cat$ is called a system of covers if
  \begin{itemize}
    \item $\mc{S}$ contains identities.
    \item $\mc{S}$ is closed under composition.
    \item $\mc{S}$ is closed under products.
    \item $\mc{S}$ is closed under weak pullback.
    \item If $f \seq \pi \in \mc{S}$, then $\pi \in \mc{S}$.
  \end{itemize}
  A pair $(\Cat, \mc{S})$ where $\Cat$ has finite products and weak pullbacks and $\mc{S}$ is a system of covers is called a
  category with covers.
  \label{def:Cover}
\end{defi}

\begin{prop}
Let $\Bicat$ be a tame cartesian bicategory and $\mc{S}$ be the class of its surjective maps. Then $(\Map(\Bicat),\mc{S})$ is a category with covers.
\end{prop}

Conversely, given a category with covers $(\Cat, \mc{S})$, we can define a tame cartesian bicategory from it,
in a way that is inspired from Corollary~\ref{cor:IneqSpan}.

\begin{defi}
Let $(\Cat, \mc{S})$ be a category with covers.
The posetal category $\SpanS{\Cat}$ is defined in analogy to $\Spantilde{\Cat}$, with a different ordering, defined as
$(X \leftarrow A \to Y) \leq (X \leftarrow B \to Y)$ if there is a commutative diagram
  
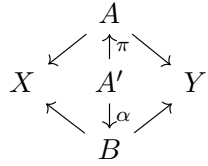
\[
  \begin{tikzcd}[sep = small]
    & \arrow[swap]{dl} A \arrow{dr} & \\
    X & A' \arrow[swap]{u}{\pi} \arrow{d}{\alpha} &  Y \\
    & \arrow{ul} B \arrow[swap]{ur} & \\
  \end{tikzcd}
  \]
with $\pi$ a cover. We define $\sim$ as the equivalence relation $s_1 \sim s_2$ if and only if $s_1 \leq s_2 \leq s_1$.
Now $\SpanS{\Cat}$ has the same objects as $\Cat$ and as morphisms $\sim$-equivalence classes of spans.
These equivalence classes are ordered via $\leq$.
Composition is given by weak pullbacks in $\Cat$. Identities, monoidal product and unit are as in $\Spantilde{\Cat}$.
\label{def:SpanS}
\end{defi}

\begin{lem}
  If $(\Cat, \mc{S})$ is a category with covers, then $\SpanS{\Cat}$ is a tame cartesian bicategory.
  \label{lem:SpanSCart}
\end{lem}
\begin{proof}
  Since the ordering in $\SpanS{\Cat}$ is finer than that in $\Spantilde{\Cat}$, it suffices to prove that the former is indeed an
  ordering and compatible with composition and monoidal product of spans. The axioms of cartesian bicategories then follow from the fact that
  $\Spantilde{\Cat}$ is a cartesian bicategory.
  \begin{itemize}
    \item Let us first check that the defined ordering is indeed reflexive and transitive.
      It is reflexive because $\mc{S}$ contains identities, it is transitive because $\mc{S}$ is closed under weak pullback.
    \item The ordering is compatible with the monoidal product of spans because the product of covers is a cover. Let us check
      that it is also compatible with composition.
      Given morphisms $A,B \from X \to Y$, $C \from Y \to Z$ such that $A \leq B$, we will prove $A \seq C \leq B \seq C$.
      The proof that given $D$ with $C \leq D$, also $B \seq C \leq B \seq D$ is similar.

      So assume $X \leftarrow A \rightarrow Y \leq X \leftarrow B \rightarrow Y$ and $Y \leftarrow C \rightarrow Z$ are given, we want to show
      that also $A \seq C \leq B \seq C$.
      By assumption, we have a commutative diagram
      
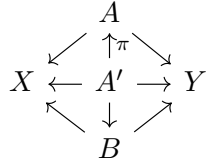
\[
        \begin{tikzcd}[sep = small]
          & A \arrow{dl} \arrow{dr} & \\
          X & \arrow{l} \arrow[swap]{u}{\pi} A' \arrow{d} \arrow{r} & Y \\
          & B \arrow{ul} \arrow{ur} &
        \end{tikzcd}
      \]
      with $\pi \in \mc{S}$.
      Now form the respective composites with $Y \leftarrow C \rightarrow Z$, so let
      
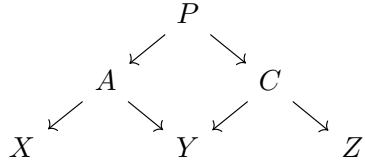
\[
        \begin{tikzcd}[sep = small]
          & & \arrow{dl} P \arrow{dr} & & \\
          & A \arrow{dl} \arrow{dr} & & C \arrow{dl} \arrow{dr} & \\
          X & & Y & & Z
        \end{tikzcd}
      \]
      and
      
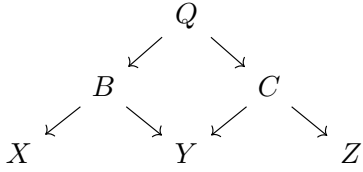
\[
        \begin{tikzcd}[sep = small]
          & & \arrow{dl} Q \arrow{dr} & & \\
          & B \arrow{dl} \arrow{dr} & & C \arrow{dl} \arrow{dr} & \\
          X & & Y & & Z
        \end{tikzcd}
      \]
      with weak pullbacks $P$ and $Q$.
      We need to construct a diagram
      
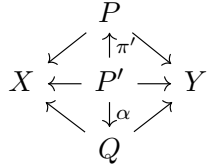
\[
        \begin{tikzcd}[sep = small]
          & P \arrow{dl} \arrow{dr} & \\
          X & \arrow{l} \arrow[swap]{u}{\pi'} P' \arrow{d}{\alpha} \arrow{r} & Y \\
          & Q \arrow{ul} \arrow{ur} &
        \end{tikzcd}
      \]
      Let $P'$ be a weak pullback
      \[
        \begin{tikzcd}[sep = small]
          P' \arrow{r}{\pi'} \arrow{d} & P \arrow{d} \\
          A' \arrow{r}{\pi} & A
        \end{tikzcd}
      \]
      Then $\pi' \in \mc{S}$ since covers are closed under weak pullback. This fits into a larger commutative diagram
      \[
        \begin{tikzcd}[sep = small]
          P' \arrow{r}{\pi'} \arrow{d} & P \arrow{d} \arrow{r} & C \arrow{dd} \\
          A' \arrow{d} \arrow{r}{\pi} & A \arrow{dr} & \\
          B \arrow{rr} & & Y
        \end{tikzcd}
      \]
      so there exists an induced morphism $\alpha \from P' \to Q$ by the weak universal property of $Q$.
      It is straightforward to check that the diagram
      \[
        \begin{tikzcd}[sep = small]
          & P \arrow{dl} \arrow{dr} & \\
          X & \arrow{l} \arrow[swap]{u}{\pi'} P' \arrow{d}{\alpha} \arrow{r} & Y \\
          & Q \arrow{ul} \arrow{ur} &
        \end{tikzcd}
      \]
      commutes.
    \item $\SpanS{\Cat}$ is furthermore tame: It is immediate to see that it has enough maps, and by definition of composition in $\SpanS{\Cat}$,
    $\stikzfig{WeakPB}$ if and only if diagram \eqref{eq:square} is a weak pullback in $\Cat$. Therefore, $\SpanS{\Cat}$ is tame.
  \end{itemize}
\end{proof}

\begin{thm}
  Let $\Bicat$ be a cartesian bicategory with enough maps and let $\mc{S}$ be the class of surjective maps.
  Then $\Bicat$ is tame if and only if
  \[
    \SpanS{\Map(\Bicat)} \cong \Bicat
  \]
  \label{thm:ReconstructTame}
\end{thm}
\begin{proof}
  \begin{itemize}
    \item If $\Bicat \cong \SpanS{\Map(\Bicat)}$, then $\Bicat$ is tame by Lemma~\ref{lem:SpanSCart}.
    \item In the other direction, Proposition~\ref{prop:MorphTame} constructs a morphism of cartesian bicategories
    $\MorphTame \from \Spantilde{\Map(\Bicat)} \to \Bicat$ which is full
    and order-preserving.
    The ordering on $\SpanS{\Map(\Bicat)}$ is a refinement of that of $\Spantilde{\Map(\Bicat)}$ and by Corollary~\ref{cor:IneqSpan}, $\MorphTame$ preserves
    this refined ordering as well, giving us an induced morphism $\MorphTame' \from \SpanS{\Map(\Bicat)} \to \Bicat$.
    This morphism is order-reflecting, again by Corollary~\ref{cor:IneqSpan}, and therefore $\MorphTame'$ is faithful by Lemma~\ref{lem:FaithfulReflectOrder}.
    Hence, $\MorphTame'$ is an isomorphism.
  \end{itemize}
\end{proof}

\begin{rem}
  There is a connection between categories with covers as defined here and sites as defined in sheaf theory,
  i.e.\ categories equipped with a Grothendieck topology~\cite{maclane2012sheaves}.
  The important difference is that a Grothendieck topology consists of families of morphisms with common codomain,
  while a cover in our sense is just a single morphism.
  However, the axioms that we require for covers somewhat mirror the axioms for Grothendieck topologies.
\end{rem}

It is tempting to say that $\Spantilde{\Cat}$ is $\SpanS{\Cat}$ with $\mc{S}$ consisting only of identity morphisms. However, identities do not form a system
of covers, because from $\pi \seq h$ being identity, it does not necessarily follow that $\pi$ is an identity, but it means that $\pi$ is a split epi.
In fact, split epis are the smallest system of covers possible.

\begin{lem}
  Split epis are the smallest system of covers.
  \label{lem:splitwit}
\end{lem}
\begin{proof}
  Let $\pi \from A \to B$ be a split epi, so there is $f \from B \to A$ such that $f \seq \pi = \id_{B}$.
  \begin{itemize}
  \item Let $\mc{S}$ be any system of covers, then
  $f \seq \pi = \id_{B} \in \mc{S}$ and therefore also $\pi \in \mc{S}$.
  \item It therefore remains to prove that split epis themselves form a system of covers.
        Almost all axioms are very straightforward to check so we will only verify the closure under weak pullbacks here.
  Let $g \from X \to B$ be any morphism and $P$ a weak pullback of $\pi$ and $g$.
  We get a commutative diagram
  \[
    \begin{tikzcd}
      B \arrow[bend left]{rr}{\id_B} \arrow{r}{f}& A \arrow{r}{\pi} & B \\
      & P \arrow{r}{\overline{\pi}} \arrow{u}{\overline{g}} & X \arrow{u}{g} \\
      X \arrow[dashed]{ur}{h} \arrow{uu}{g} \arrow[bend right]{urr}{\id_X} & &
    \end{tikzcd}
  \]
  where $h$ exists by the weak universal property of $P$, hence $\overline{\pi}$ is a split epimorphism.
  \end{itemize}
\end{proof}

\begin{lem}
  Let $\Cat$ be a category with products and weak pullbacks and let $\mc{S}$ be the class of split epis. Then
  \[
    \Spantilde{\Cat} \cong \SpanS{\Cat}
  \]
  \label{lem:CospantildeSplit}
\end{lem}
\begin{proof}
  Since $\Spantilde{\Cat}$ is tame, Theorem~\ref{thm:ReconstructTame} applies and gives $\Spantilde{\Cat} \cong \SpanS{(\Map(\Spantilde{\Cat}))}$.
  Therefore the claim follows from Proposition~\ref{prop:MapSpan}.
\end{proof}

We can use this to give a proof of Theorem~\ref{thm:Choice}. We have seen in Lemma~\ref{lem:ChoiceEntailsTame}
that enough maps together with the axiom of choice already imply tameness. Theorem~\ref{thm:Choice} will then follow from Theorem~\ref{thm:ReconstructTame}.

\begin{proof}[Proof of Theorem~\ref{thm:Choice}]
  Let $\Bicat$ be a cartesian bicategory with enough maps.
  \begin{itemize}
    \item If $\Bicat \cong \Spantilde{\Map(\Bicat)}$ then $\Bicat$ satisfies~\eqref{eq:AC} by Corollary~\ref{cor:spanAC}.
    \item Conversely, if $\Bicat$ satisfies~\eqref{eq:AC}, then by Lemma~\ref{lem:ChoiceEntailsTame}, $\Bicat$ is tame.
    Therefore Theorem~\ref{thm:ReconstructTame} applies and we get
    \[
      \Bicat \cong \SpanS{\Map(\Bicat)}
    \]
    where $\mc{S}$ is the class of surjective maps in $\Bicat$. Since $\Bicat$ satisfies~\eqref{eq:AC}, by Lemma~\ref{lem:ChoiceSplitsEpis}
    the class of surjectives $\mc{S}$ agrees with the class of split epis.
    Therefore we get $\SpanS{\Map(\Bicat)} \cong \Spantilde{\Map(\Bicat)}$ by Lemma~\ref{lem:CospantildeSplit}, which finishes the proof.
  \end{itemize}
\end{proof}

\begin{exa}
  Recall from Example~\ref{ex:Frob} the category $\FinSet$ of finite sets. This category has the interesting property that injective functions are
  closed under pushout. In fact, the class $\mc{I}$ of injective functions satisfies the dual of all properties of Definition~\ref{def:Cover}, in
  other words it is a system of covers on the opposite category $\op{\FinSet}$. The corresponding tame cartesian bicategory $\SpanI{\op{\FinSet}}$
  turns out to be $\mathbf{ERel}$, the cartesian bicategory of equivalence relations.
  \label{ex:ERelSpan}
\end{exa}

\section{Related work}\label{sec:relatedWork}

Another common example of cartesian bicategories, considered in \cite{carboni1987cartesian}, is the category of relations of
a regular category. The following definitions can be found in~\cite{butz1998regular}.

\begin{defi}
  Let $\mathcal{C}$ be a category. A kernel pair of a morphism $f \from X \to Y$ is a pair of

  $p_1, p_2 \from P \to X$ such that the diagram
  \[
    \begin{tikzcd}[sep=small]
      P \arrow{r}{p_1} \arrow[swap]{d}{p_2} & X \arrow{d}{f} \\
      X \arrow{r}{f} & Y
    \end{tikzcd}
  \]
  is a pullback. An epimorphism is regular if it is the coequaliser of some pair of morphisms.
  $\mathcal{C}$ is regular if it has finite limits, coequalisers of kernel pairs and regular epis are stable under pullback.
\end{defi}

Regular categories admit a well-behaved factorisation system, where every morphism factors as a
regular epi followed by a mono. The factorisation is used to define the cartesian bicategory of relations of a regular category.

\begin{defi}
  Given a regular category $\mathcal{C}$, let $\Relc{\mathcal{C}}$ be the category with the same objects as $\mathcal{C}$ and morphisms
  $X \to Y$ jointly mono spans, i.e.\ spans $X \xleftarrow{f} A \xrightarrow{g} Y$ such that the induced map $A \xrightarrow{\langle f, g \rangle} X \times Y$ is mono. For an arbitrary span,  $X \xleftarrow{f} A \xrightarrow{g} Y$, its \emph{image} is the jointly mono span given by taking the regular epi-mono factorisation of $A \xrightarrow{\langle f, g \rangle} X \times Y$.
  The composition of two jointly mono spans is given by first composing them as spans via pullback and then taking the image of the resulting span. The identity $X \to X$ is given by the jointly mono span $X \xleftarrow{\id_X} X \xrightarrow{\id_X} X$.
  Similar to $\Spantilde{\Cat}$, the categorical product of $\Cat$ induces a monoidal product on $\Relc{\Cat}$.
  Furthermore, the ordering is defined as for $\Spantilde{\Cat}$:
$(X \leftarrow A \to Y) \leq (X \leftarrow B \to Y)$ if there exists a morphism of spans
  $\alpha \from (X \leftarrow A \to Y) \Rightarrow (X \leftarrow B \to Y)$.
\end{defi}

This is the classical approach to regular categories. But they also fit nicely into the context of tame cartesian bicategories, because,
as it turns out, the class of regular epis is a system of covers. Fix a regular category $\Cat$,
then it is well known that identities are regular epis, in fact all split epis are, and that the class of regular epis is closed under
composition, product, pullbacks and whenever $f \seq \pi$ is a regular epi, so is $\pi$. From the stability under pullback it is easy to deduce
stability under weak pullback, as follows: If
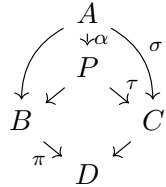
\[
  \begin{tikzcd}[sep = tiny]
    & \arrow[swap, bend right]{ddl} A \arrow[dashed]{d}{\alpha} \arrow[bend left]{ddr}{\sigma} & \\
    & \arrow[swap]{dl} P \arrow{dr}{\tau} & \\
    B \arrow[swap]{dr}{\pi} & & \arrow{dl} C \\
    & D & \\
  \end{tikzcd}
\]
is a commutative diagram where $P$ is a pullback, $A$ a weak pullback and $\alpha$ the induced morphism from the universal property of $P$.
If $\pi$ is a regular epi, then so is $\tau$ because regular epis are stable under pullback and $\alpha$ is easily seen to be a split epi by the
universal property of $P$. Therefore $\tau$ is a regular epi as well. In other words, the class of regular epis form a system of covers.
This can be used to characterise $\Relc{\Cat}$ as follows:

\begin{lem}
  Let $\Cat$ be a regular category. Let $\mc{S}$ be the class of regular epis, then
  \[ \SpanS{\Cat} \cong \Relc{\Cat} \]
\end{lem}
\begin{proof}
  There is a natural morphism of cartesian bicategories $F \from \Span{\Cat} \to \Relc{\Cat}$ that sends a span to its image.
  Whenever
  
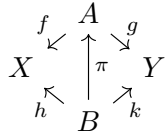
\[
  \begin{tikzcd}[sep = tiny]
    & \arrow[swap]{dl}{f} A \arrow{dr}{g} & \\
    X & & Y \\
    & \arrow{ul}{h} B \arrow[swap]{uu}{\pi} \arrow[swap]{ur}{k} & \\
  \end{tikzcd}
  \]
  is a commutative diagram with $\pi$ a regular epi then the top and the bottom span will have the same image. Therefore $F$ is compatible with the
  ordering on $\SpanS{\Cat}$, so it induces a morphism $F' \from \SpanS{\Cat} \to \Relc{\Cat}$.
  It is identity on objects, easily seen to be full and furthermore faithful because in $\SpanS{\Cat}$, every span is equal to its image.
\end{proof}

It is known that surjective maps in $\Relc{\Cat}$ are precisely regular epis in $\Cat$, see~\cite[Theorem 3.5]{carboni1987cartesian}. Using Proposition~\ref{prop:ACiffEpisSplit}, we have the following.

\begin{cor}
$\Spantilde{\Cat} \cong \Relc{\Cat}$ iff regular epis split in $\Cat$.
\end{cor}

\bibliographystyle{plainurl}
\bibliography{Bibliography}

\end{document}